\documentclass{llncs}

\newcommand{\knip}[1]{}
\usepackage{url}

\usepackage{graphicx}
\usepackage{boxedminipage}

\newtheorem{observation}{Observation}

\newcommand{\sP}{{\sf P}}
\newcommand{\NP}{{\sf NP}}

\newcommand{\ie}{i.e.~}
\newcommand{\eg}{e.g.~}

\newcommand{\mc}{\mathcal}

\newcommand{\marge}[1]{}

\DeclareSymbolFont{AMSb}{U}{msb}{m}{n}
\DeclareSymbolFontAlphabet{\mathbb}{AMSb}

\newcommand{\problemIDP}{\textsc{Induced Disjoint Paths}}
\newcommand{\problemRIDP}{\textsc{Requirement Induced Disjoint Paths}}
\newcommand{\problemDP}{\textsc{Disjoint Paths}}
\newcommand{\problemMCIS}{\textsc{Multicolored Independent Set}}

\begin{document}

\title{Induced Disjoint Paths in\\ Circular-Arc Graphs in Linear Time
\thanks{This work is supported by EPSRC (EP/K025090/1) and Royal Society (JP100692).
The research leading to these results has also received funding from the European Research Council under the European Union's Seventh Framework Programme (FP/2007-2013) / ERC Grant Agreement n.\ 267959.}
}
\author{
Petr A.~Golovach\inst{1}
\and Dani\"el Paulusma\inst{2}
\and Erik Jan van Leeuwen\inst{3}
\institute{Department of informatics, University of Bergen, Norway\\
\email{petr.golovach@ii.uib.no}
\and
School of Engineering and Computer Science, Durham University, UK\\ 
\email{daniel.paulusma@durham.ac.uk}
\and 
Max-Planck Institut f\"{u}r Informatik, Saarbr\"{u}cken, Germany,
\email{erikjan@mpi-inf.mpg.de}
}}
\maketitle
\thispagestyle{plain}
\setcounter{footnote}{0}

\begin{abstract}
The \problemIDP{} problem is to test whether a graph $G$ with $k$ distinct pairs of vertices $(s_{i},t_{i})$ contains paths $P_{1},\ldots,P_{k}$ such that $P_{i}$ connects $s_{i}$ and $t_{i}$ for $i=1,\ldots,k$, and $P_{i}$ and $P_{j}$ have neither common vertices nor adjacent vertices (except perhaps their ends) for $1 \leq i < j \leq k$.
We present a linear-time algorithm for \problemIDP{} on circular-arc graphs.
For interval graphs, we exhibit a linear-time algorithm for the generalization of \problemIDP{} where the pairs $(s_{i},t_{i})$ are not necessarily distinct.
\end{abstract}

\section{Introduction}\label{s-intro}
A classic algorithmic problem on a graph $G$ with $k$ distinct pairs of vertices $(s_{i},t_{i})$ is to find vertex-disjoint
\footnote{There is also a version of the problem in which the paths are required to be edge-disjoint. We do not consider that version in this paper.} 
paths $P_{1},\ldots,P_{k}$ such that $P_{i}$ connects $s_{i}$ and $t_{i}$. Known as the \problemDP{} problem, it is \NP-complete on general graphs~\cite{Ka75}, but can be solved in $O(n^{3})$ time for any fixed integer $k$~\cite{RS95} (\ie it is fixed-parameter tractable). A generalization of this problem is \problemIDP{}: given $k$ distinct pairs of vertices $(s_{i},t_{i})$ in a graph $G$, find paths $P_{1},\ldots,P_{k}$ such that $P_{i}$ connects $s_{i}$ and $t_{i}$ for $i=1,\ldots,k$ and the paths are \emph{mutually induced}, that is, no two paths $P_{i},P_{j}$ have common or adjacent vertices (except perhaps their end-vertices). The \problemIDP{} problem indeed generalizes the \problemDP{} problem, since the latter can be reduced to the former by subdividing every edge of the graph. This makes the problem much harder: \problemIDP{} is \NP-complete even for instances with $k=2$~\cite{Bi91,Fe89}, and thus in particular is not fixed-parameter tractable unless \sP=\NP.

The hardness of both \problemDP{} and \problemIDP{} on general graphs inspired research on their complexity on structured graph classes. On the negative side, \problemDP{} remains \NP-complete on line graphs~\cite{Ly75} and split graphs~\cite{HvtHSvL2014}, \problemIDP{} remains \NP-complete on claw-free graphs~\cite{FKLP12}, and both problems remain \NP-complete on planar graphs~\cite{KvL84,GPV12}. In these cases, however, fixed-parameter algorithms are known~\cite{GPV12b,HvtHSvL2014,KK12,RRSS93,RS95}. On the positive side, polynomial-time algorithms for \problemDP{} exist on graphs of bounded treewidth~\cite{Reed1997} and graphs of cliquewidth at most~$2$~\cite{GurWan2006}, and for \problemIDP{} on AT-free graphs~\cite{GPV12} and chordal graphs~\cite{BGHHKP12}.

We focus on the complexity of \problemIDP{} on circular-arc graphs. Recall that a \emph{circular-arc graph} $G$ has a \emph{representation} in which each vertex of $G$ corresponds to an arc of a circle, and two vertices of $G$ are adjacent if and only if their corresponding arcs intersect. Circular-arc graphs generalize \emph{interval graphs}, which have a representation in which each vertex corresponds to an interval of the line, and two vertices are adjacent if and only if their corresponding intervals intersect. The complexity of \problemDP{} is known: it is \NP-complete already on interval graphs~\cite{NS1996}. In contrast, for \problemIDP{}, the authors of the present work recently showed a polynomial-time algorithm on circular-arc graphs~\cite{GPV12b}, and a polynomial-time algorithm on interval graphs is implied by that work, as well as by the polynomial-time algorithms on AT-free graphs~\cite{GPV12} and chordal graphs~\cite{BGHHKP12}. These algorithms, however, do not fully settle the complexity of \problemIDP{} on circular-arc graphs (and interval graphs), because the question whether a linear-time algorithm exists has been left open.

In this paper, we exhibit a linear-time algorithm for \problemIDP{} on circular-arc graphs. This improves on the known algorithm on circular-arc graphs as well as the known algorithms for interval graphs. We also introduce a generalization of \problemIDP{} called \problemRIDP{}, which is to find $r_{i}$ paths that connect $s_{i}$ and $t_{i}$ for $i=1,\ldots,k$, such that all paths are mutually induced. We present a linear-time algorithm for \problemRIDP{} on interval graphs. To solve these problems, our algorithms first preprocesses the instance. Some of the preprocessing rules build on our earlier work on \problemIDP{}~\cite{GPV12,GPV12b}, but special care is required to adapt them for \problemRIDP{} and to execute them in linear time. Most preprocessing rules, however, are novel. After the preprocessing stage, the algorithms identify a set of candidate paths for each pair $(s_{i},t_{i})$. For each candidate path for a pair $(s_{i},t_{i})$, we add an arc with color $i$ that corresponds to the path to an auxiliary graph. Finally, we show that it suffices to find an independent set in this auxiliary graph that contains $r_{i}$ arcs of each color. We show that the algorithms perform all stages in linear time.

\section{Preliminaries}\label{s-pre}
We only consider finite undirected graphs that have no loops and no multiple edges. 
We refer to the textbook of Diestel~\cite{Di05} for any standard graph terminology not defined here. Let $G=(V,E)$ be a graph. 
For a set $S\subseteq V$, the graph $G[S]$ denotes the subgraph of $G$ {\it induced by} $S$, that is, the graph with vertex set $S$ and edge set $\{uv \in E \mid  u,v\in S\}$.  We write $G-S=G[V\setminus S]$. 
We denote the (open) neighborhood of a vertex $u$ by $N_G(u)=\{v \mid uv\in E\}$ and its closed neighborhood by $N_G[u]=N_G(u)\cup \{u\}$.
We denote the neighborhood of a set $U\subseteq V$ by $N_G(U)=\{v\in V\setminus U \mid  uv\in E\; \mbox{for some}\; u\in U\}$ and $N_G[U]=U\cup N_G(U)$. 
We denote the degree of a vertex $u$ by $\deg_G(u)=|N_G(u)|$.

We denote an unordered pair of elements $x,y$ by $\{x,y\}$  (\ie $\{x,y\}=\{y,x\}$).

\paragraph{Problem Definition}
Let $P=v_1\cdots v_r$ be a path (we call such a path a {\em $v_1v_r$-path}).
The vertices $v_1$ and $v_r$  are the {\it ends} or {\it end-vertices} of $P$, and the vertices $v_2,\ldots,v_{r-1}$ are the \emph{inner vertices} of $P$.
We say that an edge $v_iv_j$, $i+1<j$, is an \emph{inner chord} of $P$ if $v_i$ or $v_j$ is an inner vertex of $P$.
Distinct paths $P_1,\ldots, P_{\ell}$ in a graph $G$  are {\it mutually induced} if:
\begin{itemize}
\item [(i)] each $P_i$ has no inner chords;
\item [(ii)]  any distinct $P_i,P_j$ may only share vertices that are ends of both paths;
\item [(iii)]  no inner vertex $u$ of any $P_i$ is adjacent to a vertex $v$ of some $P_j$ for $j\neq i$, except when $v$ is an end-vertex of both $P_i$ and $P_{j}$.
\end{itemize}
Notice that condition~(i) may be assumed without loss of generality. This definition is more general than the definition in Section~\ref{s-intro}, as it allows the end-vertices of distinct paths to be the same or adjacent. 
We can now formally state our decision problem (where a {\em terminal} is some specified vertex).
\begin{center}
\begin{boxedminipage}{1\textwidth}
\ \problemRIDP{}\\[3mm]
\begin{tabular}{ r l }
\textit{~~~~Instance:} & a graph $G$, 
 $k$  pairs of distinct terminals $(s_1,t_1),\ldots,(s_k,t_k)$ such\\ 
         & that $\{s_i,t_i\}\neq\{s_j,t_j\}$ for $0\leq i<j\leq k$, and $k$ positive integers\\
        & $r_1,\ldots,r_k$.\\
\textit{Question:} & does $G$ have  $\ell=r_1+\ldots+r_k$ mutually induced  paths $P_1,\ldots,P_\ell$ \\&
        such that exactly $r_i$ of these paths join $s_i$ and $t_i$ for $1\leq i\leq k$?  \\
\end{tabular}
\end{boxedminipage}
\end{center}
If $r_1=\ldots=r_k=1$, then  the problem is called \problemIDP{}. 
The paths $P_1,\ldots,P_\ell$ are said to form a {\em solution} for a given instance, and we call every such path a {\em solution path}.

The problem definition allows a vertex $v$ to be a terminal in two or more pairs $(s_i,t_i)$ and $(s_j,t_j)$. For instance, $v=s_i=s_j$ is possible.
This corresponds to property (ii) of our definition of ``being mutually induced''. 
In order to avoid any confusion, we will view $s_i$ and $s_j$ as two different terminals ``placed on'' vertex~$v$.
Formally, we call $v$ a {\em terminal vertex} that
{\em represents} a terminal $s_i$ or $t_i$ if $u=s_i$ or $u=t_i$, respectively.
We let $T_v$  denote the set of terminals represented by~$v$. If $T_v=\emptyset$, we call $v$ a {\em non-terminal} vertex. 
We say that the two terminals $s_i$ and~$t_i$ of a terminal pair $(s_i,t_i)$ are {\em partners} of each other. 
If $s_i$ is represented by $u$ and $t_i$ by $v$, then we also call a $uv$-path an $s_it_i$-path.
By our problem definition, each terminal pair $(s_i,t_i)$ consists of two distinct terminals. Hence, two partners are never represented by the same vertex.

By Property~(i), each solution path $P$  has no inner chords. It is an induced path if and only if its ends are non-adjacent.
If two adjacent vertices $u$ and $v$ represent terminals vertices belonging to the same pair $(s_i,t_i)$, then the path $uv$ is called
a {\em terminal path} for $s_i$, $t_i$.
We need the following observation.

\begin{observation}\label{o-trivial}
Any yes-instance of \problemRIDP{} has a solution that contains 
all terminal paths. In particular, a terminal path for a pair $(s_i,t_i)$ is the unique $s_it_i$-path in this solution if $r_i=1$.
\end{observation}

\paragraph{Graph Classes}
Recall the definition of circular-arc and interval graphs from the introduction. Both graph types can be recognized in linear time and a corresponding representation can be found in linear time:

\begin{theorem}[\cite{BoothL76}, see also~\cite{HabibMCPV00,KorteM89}]\label{t-bl}
An interval graph $G$ with $n$ vertices and $m$ edges can be recognized in $O(n+m)$ time. In the same time, a representation of $G$ can be constructed with interval end-points 
$1,\ldots,2n$.   
\end{theorem}

The first linear-time recognition algorithm for circular-arc graphs was given by McConnell~\cite{M2003} (see also~\cite{KaplanN11}).

\begin{theorem}[\cite{M2003}]\label{t-mcc}
A circular-arc graph $G$ with $n$ vertices and $m$ edges can be recognized in $O(n+m)$ time. In the same time, a representation of $G$ can be constructed with arc end-points 
clockwise enumerated as $1,\ldots,2n$.   
\end{theorem}

By Theorems~\ref{t-bl} and \ref{t-mcc}, we always assume that an interval or circular-arc graph is given both by its adjacency list and its representation. Moreover, we assume that all the end-points of the intervals/arcs in the representation are distinct integers $1,\ldots,2n$. 
Notice that using a representation we can check adjacency in $O(1)$ time. 
By slight abuse of notation, we often do not distinguish between the vertices and their corresponding intervals/arcs, \eg we may speak of terminal intervals/arcs instead of terminal vertices.

For a vertex $u$ of an interval graph, $l_u$ and $r_u$ denote the left and right end-point of $u$, respectively; note that the degree of $u$ is at least $(r_u-l_u-1)/2$.
For circular-arc graphs, we equate ``left'' to ``counterclockwise'' and ``right'' to ``clockwise''. Then, in the same way as for interval graphs,  we let $l_u$ and $r_u$ denote the left and right end-point of a vertex $u$, respectively. In this way we are able to define similar terminology for both interval and circular-arc graphs. 
For two points $x,y$ on the line or circle, we write $x\leq y$ if $y$ lies to the right with respect to $x$, and $x<y$ if $x\leq  y$ and $x\neq y$.
We say that a point $z$ \emph{lies between} points $x$ and $y$, if $x\leq z\leq y$. 
We say that a vertex $u$ {\em lies between} points $x$ and $y$ if $x\leq l_u < r_u \leq y$ (recall that $l_u$ and $r_u$ are distinct integers).
Finally, a vertex $u$ \emph{lies between} two other vertices $v,w$ if it lies between $r_v$ and $l_w$; note that 
in that case we have in fact that $r_v< l_u< r_u< l_w$ by our assumption on the interval representation.

An {\it independent set} in a graph $G$ is a set of vertices that are pairwise non-adjacent.  
At some stage, our algorithm for \problemIDP{} on circular-arc graphs needs to compute a largest independent set of a circular-arc graph. This takes linear time:

\begin{theorem}[\cite{GolumbicH88}]\label{t-maxin}
If the arc end-points of a circular-arc graph $G$ are sorted,  then
a largest independent set of $G$ can be found in $O(n)$ time.
\end{theorem}

\section{Interval Graphs}\label{s-int}
In this section we develop a linear-time algorithm that solves \problemRIDP{} on interval graphs. A possible approach would be the following greedy algorithm: find a terminal vertex with the leftmost right end-point, trace path(s) for the corresponding terminal pairs, greedily choose the non-terminal vertex with the leftmost right end-point that does not create conflicts with vertices already chosen, and proceed in a greedy way. However, we do not elaborate on this approach for two reasons. Firstly, this approach would require a thorough case analysis (just like our algorithm, and thus not be substantially simpler). Secondly, and more importantly, the goal of this paper is to design a linear-time algorithm for \problemIDP{} on circular-arc graphs, where we have no natural starting point for a similar greedy approach and guessing such a starting point would irrevocably lead to a quadratic-time algorithm. Therefore, we present a different approach already for interval graphs.

We describe the main constructs of our algorithm. Consider an instance  of \problemRIDP{}.
Let $P$ be an $s_it_i$-path that is not a terminal path, \ie that has at least one inner vertex. 
Let $I_P$ be the interval on the line obtained by taking the union of the intervals that correspond to the inner vertices of $P$. We say that $P$ \emph{covers} the interval $I_P$. Because $P$ is an $s_it_i$-path, we say that $I_{P}$ has color $i$.

\begin{lemma}\label{l-cover}
Let $P_1,\ldots, P_\ell$ form a solution. The following statements hold:
\begin{itemize}
\item[i)] For $1\leq i\leq k$, any interval $I_{P_a}$ with color $i$ intersects the intervals that represent $s_{i}$ and $t_{i}$ and does not intersect any other terminal interval;
\item[ii)] For $1\leq a<b\leq \ell$, $I_{P_a}\cap I_{P_b}=\emptyset$;
\item[iii)] For $1\leq i<j \leq k$, there is no interval with color $j$ that lies between two intervals with color $i$, or vice versa.
\end{itemize}
\end{lemma}

\begin{proof}
Properties i) and ii) follow immediately from definition. In order to show iii), assume that an interval $I_{P_c}$ with color $j$
lies between two intervals $I_{P_a}$ and $I_{P_b}$, both with color $i$, for some $i,j$ with $i\neq j$.
Let $u$ and $v$ represent $s_i$ and~$t_i$. By i), $I_{P_a}$ and $I_{P_b}$ each intersect $u$ and~$v$. 
Then $I_{P_c}$ also intersects $u$ and~$v$. As $i\neq j$, we find that $u$ or $v$ represents neither 
$s_j$ nor $t_j$, contradicting~i).\qed  
\end{proof}

We now outline our algorithm. Following Observation~\ref{o-trivial}, we take all terminal paths into the solution. This might reduce the requirement $r_{i}$ by $1$ for some $i$. To find the remaining  paths for all $i$, we determine a set of ``candidate paths'' that might or might not be used in the solution that we are constructing. The set of candidate paths is constructed such that for any $s_it_i$ solution path $P$ there is a candidate path $P'$ such that $P'$ is also an $s_it_i$-path and $I_{P'} \subseteq I_{P}$. We guarantee that the set of candidate paths has size $O(n)$. By Lemma~\ref{l-cover}, the paths that are selected in a solution must cover distinct parts of the line. Therefore, we create an auxiliary interval graph $H$ that consists of all intervals covered by the candidate paths. The intervals covered by candidate $s_it_i$-paths all receive color $i$, for $i=1,\ldots,k$. It then suffices to find an independent set with the required number of vertices of each color in $H$. 

In the remainder of this section, we describe all steps of the algorithm in detail. 
We say that a step is {\em safe} if it runs in time $O(n+m+k)$ and is correct the following sense: 
\begin{itemize}
\item [(i)] a \texttt{No}-answer is given for no-instances only;
\item [(ii)] if a new instance is obtained, then it has a solution if and only if the original instance has so.
\item [(iii)] if a set of intervals that are all colored with color~$i$ is added to $H$, then this set has size $O(n)$ and corresponds to a \emph{candidate set} of {\em candidate paths}.
\end{itemize}
The algorithm assumes that an interval representation of $G$ is known, as given by Theorem~\ref{t-bl}. It also maintains an auxiliary interval graph $H$, initially empty. Recall that any vertex that we add to $H$ will correspond to a candidate path for a solution. While adding vertices to $H$, we maintain an interval representation of $H$. Finally, the algorithm maintains a set $\mc{P}$ of paths, initially empty, which will form a solution for the instance (should it be a yes-instance). We let $T = \{s_{1},t_{1},\ldots,s_{k},t_{k}\}$ be the set of all terminals. A terminal pair $(s_i,t_i)$ is a {\em multi-pair} if $r_i\geq 2$, and a {\em simple pair} otherwise. 
The algorithm roughly consists of three stages: preprocess, construct $H$, and find an independent set.

\subsection{Stage I: Preprocess}
The only operations performed on $G$ by our algorithm are vertex deletions. Hence, the graph that we obtain after each step is still interval. For simplicity, we denote this graph by $G$ as well. 

\medskip
\noindent
{\bf Step~1.} Delete all non-terminal vertices that are adjacent to at least three terminal vertices. 

\begin{lemma}\label{l-1}
Step 1 is safe.
\end{lemma}

\begin{proof}
Any internal vertex of a path of a solution is adjacent to at most two terminal vertices, which are the end-vertices of the path. Hence, any non-terminal vertex that is adjacent to at least three terminal vertices  cannot be used in any solution. Therefore, Step~1 is correct. In $O(n+m)$ time, we can check the neighborhood of each non-terminal vertex through the adjacency list and count the number of terminals.
\qed
\end{proof}

\medskip
\noindent
{\bf Step 2.} Check if there is a multi-pair that is represented by two non-adjacent terminal vertices. If so, then return a \texttt{No}-answer.

\begin{lemma}\label{l-2}
Step 2 is safe.
\end{lemma}

\begin{proof}
Step~2 is correct, because there must exist at least two solution paths between the terminal vertices of a multi-pair. If the two terminal vertices are not adjacent, the union of the vertices of these two paths induces a cycle on at least four vertices in~$G$. This is not possible in an interval graph.
Using the list of terminal pairs, Step~2 takes $O(k)$ time.\qed
\end{proof}

Suppose that we have not returned a \texttt{No}-answer after performing Step~2.
In the next step, for each multi-pair, we identify a set of paths that together with the terminal paths form all candidate paths.

\medskip
\noindent
{\bf Step~3.} For each non-terminal vertex $u$  adjacent to terminal vertices $v$ and $w$  representing multi-pair terminals $s_i$ and $t_i$, add $I_{vuw}$ with color $i$ to $V_{H}$, and delete $u$ from $G$.

\begin{lemma}\label{l-3}
Step 3 is safe. Moreover, for any multi-pair $(s_i,t_i)$, if $P$ is a solution $s_it_i$-path with at least one inner vertex, then there is a candidate $s_it_i$-path $P'$ with $I_{P'} \subseteq I_{P}$.
\end{lemma}

\begin{proof}
We first prove that Step~3 is correct. Let $u$ be a non-terminal vertex 
adjacent to terminal vertices $v$ and $w$  representing terminals $s_i$ and $t_i$ from a multi-pair $(s_i,t_i)$.
By Lemma~\ref{l-1}, we find that $u$ is not adjacent to any other terminal vertices. Hence, $vuw$ may be considered as 
a candidate path for a solution. Moreover, because $u$ is adjacent to both $v$ and $w$, we deduce the following.
Firstly, every $s_it_i$-path in a solution has at most one inner vertex; otherwise its vertices would induce a cycle on at least four 
vertices in $G$, as $v,w$ are adjacent by Step~2. Hence, the set of intervals added to $V_H$ for each multi-pair $(s_i,t_i)$ contains all
possible solution paths for $(s_i,t_i)$, and as such corresponds to a candidate set for $(s_i,t_i)$.
Secondly, $u$ may not be used
in a solution path for a terminal pair $(s_j,t_j)$ with $j\neq i$. Hence, we can safely remove $u$ from $G$.
Because we only added intervals to $H$ that correspond to distinct vertices,  we added $O(n)$ vertices to $V_H$ in total.

We now show how to perform Step~3 in $O(n+m+k)$ time. 
Construct $2n$ buckets $B_1,\ldots,B_n$. We add every vertex $u\in V_G$ to buckets $B_{l_u},\ldots,B_{r_u}$. 
By the definition of our interval representation, the degree of $u$ in $G$ is equal to~$r_u-l_u-1$. 
Hence, $|B_1|+\ldots +|B_n|\leq \sum_{u\in V_G}(r_u-l_u+1)\leq\sum_{u\in V_G}(2\deg_G(u)+2)=4m+2n$, implying that filling the buckets takes $O(n+m)$ time in total.
For any terminal intervals $v$ and $w$ that represent terminals $s_i$ and $t_i$ of a multi-pair, determine the intersection interval $[l,r]$ of $v$ and $w$ (by Step~2, $v$ and $w$ are adjacent). Then remove every vertex $u$ of $G$ that is in $B_l\cup \cdots \cup B_r$, color $I_{vuw}$ with color~$i$, and add $I_{vuw}$ to $V_H$. 
This takes time $O(n+m+k)$ in total, and $O(n)$ intervals are added to $H$. 
\qed
\end{proof}

In the next two steps, which are inspired by our earlier work on \problemIDP~\cite{GPV12,GPV12b}, we get rid of all adjacent terminal vertices that represent the same terminal pair. This includes (but is not limited to) all multi-pairs.

\medskip
\noindent
{\bf Step 4.}
Find the set $Z$ of all terminal vertices $v$ such that $v$ only represents terminals  whose partners are in $N_G(v)$.
Delete the vertices of $Z$ and all non-terminal vertices of $N_G(Z)$ from $G$. Delete from $T$ the terminals of all terminal pairs $(s_i,t_i)$ with $s_i\in T_v$ or $t_i\in T_v$ for some $v\in Z$. Put all terminal paths corresponding to deleted terminal pairs in ${\cal P}$.

\begin{lemma}\label{l-4}
Step 4 is safe.
\end{lemma}

\begin{proof}
We first show that Step~4 is correct. Let $\{s_{i_1},\ldots,s_{i_p},t_{j_1},\ldots,t_{j_q}\}$ be the union of all terminals represented by
vertices in  $Z$. 
By Observation~\ref{o-trivial}, we may assume that each terminal path for $(s_{i_a},t_{i_a})$ for $a=1,\ldots,p$ and each terminal path for 
$(s_{j_b},t_{j_b})$ for $b = 1,\ldots,q$ is 
in a solution, if our instance is a yes-instance. Hence, we can safely put these terminal paths in ${\cal P}$. Moreover, as we already identified a candidate set for all multi-pairs  in Step~3, we may safely remove each of the two terminals of every pair $(s_{i_a},t_{i_a})$ for $a=1,\ldots,p$ and every pair $(s_{j_b},t_{j_b})$ for $b=1,\ldots,q$
from $T$.

Let $u$ be a non-terminal vertex in $N_G(Z)$. 
Then $u$ is not adjacent to two terminal vertices representing two terminals from a multi-pair, as otherwise we would
have removed $u$ in Step~3 already. Moreover, $u$ is not used as an inner vertex of a solution path for a simple terminal pair $(s_i,t_i)$ either, for the following two reasons.
Firstly, if $s_i$ or $t_i$ is represented by a vertex in~$Z$, we would use the corresponding terminal path for a solution due to
Observation~\ref{o-trivial}. Secondly, if both $s_i$ and $t_i$ are not represented by a vertex in~$Z$,  we could still not use
$u$ as an inner vertex for an $s_it_i$-path, as $u$ is adjacent to some terminal vertex in~$Z$.

We now show how to perform Step~4 in $O(n+m+k)$ time.
We ``mark'' each terminal vertex. Then we go through the list of terminal pairs, and if a pair $(s_{i},t_{i})$ is not represented by adjacent terminal vertices, then we ``unmark''  these terminal vertices. 
The set $Z$ is the set of  all ``marked''  terminal vertices that are left in the end. 
By using the interval representation, obtaining $Z$ takes $O(k)$ time.
By using the adjacency lists of the vertices of $Z$, we find all non-terminal vertices of $N_G(Z)$. Each time we find such a non-terminal vertex, we delete it from $G$. Afterward, we delete all vertices of $Z$.
This takes $O(n+m)$ time.
Finally, we go through the list of terminal pairs, and if a terminal $s_i$ or $t_i$ is in $Z$, we delete both $s_i$ and $t_i$ from $T$ and add its terminal path to $\mc{P}$.
This takes $O(k)$ time. We conclude that
the total running time of performing Step~4 is~$O(n+m+k)$.\qed
\end{proof}

After Step~4, each terminal vertex represents at least one terminal whose partner is at distance at least~2. There may still be terminal pairs whose terminals are represented by adjacent vertices. We deal with such pairs in the next step.

\medskip
\noindent
{\bf Step 5.}
Delete all terminals $s_i$ and $t_i$ represented by adjacent terminal vertices from the terminal list, and delete all common non-terminal neighbors of the terminal vertices that represent $s_i$ and $t_i$.
Put all terminal paths corresponding to deleted terminals in ${\cal P}$.

\begin{lemma}\label{l-5}
Step 5 is safe.
\end{lemma}

\begin{proof}
By using the interval representation, Step~5 can be done in $O(n+m+k)$ time. Hence, it remains to show that Step~5 is correct.

First, we may assume without loss of generality that a solution contains all terminal paths by Observation~\ref{o-trivial}.
Hence, we may safely put these terminal paths in ${\cal P}$, and delete terminals that are represented by adjacent terminal vertices if $(s_i,t_i)$ is not a multi-pair; if $(s_i,t_i)$ is a multi-pair, then all candidate paths have already been identified in Step~3, and thus $s_i$ and $t_i$ may be deleted as well. 

Second, if a solution path contains an inner vertex $u$ adjacent to a terminal vertex $v$ representing a terminal that we remove in Step~5, then the reason is that $u$ belongs to a solution path for a terminal pair $(s_j,t_j)$ where $s_j$ or $t_j$ is 
represented by $v$ as well (note that $v$ represents at least one terminal whose partner is not represented by a neighbor of $v$, as otherwise we would have removed $v$ in Step~4). Hence, $u$ is allowed to be adjacent to $v$ by definition, except if $u$ is adjacent to both the terminal vertex that represents $s_i$ and the terminal vertex that represents $t_i$. Since these common neighbors are removed in Step~5, however, this is not possible.
\qed
\end{proof}

Call a terminal pair {\it long} if its two terminals are represented by vertices of distance at least~2. 
After Step~5, all terminal pairs are long. Therefore, by Step~2, there are no multi-pairs anymore.
Assume that there are $k'\leq k$ terminal pairs left; note that $k'=0$ is possible.

\medskip
\noindent
{\bf Step 6.}
Check if there exists a terminal vertex that represents three or more terminals. If so, then return a \texttt{No}-answer.

\begin{lemma}\label{l-6}
Step~6 is safe.
\end{lemma}

\begin{proof}
We first prove that Step~6 is correct.
For contradiction, assume that a terminal vertex $u$ represents at least three terminals $s_h,s_i,s_j$.
Due to Step~5, these terminals belong to long pairs. 
Let $v_1,v_2,v_3$ denote the terminal vertices that represent $t_h$, $t_i$, $t_j$, respectively. 
Because $u$ is not adjacent to any of $v_1,v_2,v_3$, every solution has $s_ht_h$, $s_it_i$, and $s_jt_j$-paths that
each contain at least one inner vertex $x_1,x_2,x_3$, respectively. Assume without loss of generality that $x_1$, $x_2$, $x_3$
are  adjacent to $u$.  The intervals $x_1,x_2,x_3$ do not intersect each other but they do intersect $u$. 
Assume without loss of generality that $x_2$ lies between $x_1$ and $x_3$. 
Then all the vertices of the $s_it_i$-path except $u$ lie between $x_1$ and $x_3$. Therefore, $u$ and $v_2$ are adjacent.
This contradicts with the fact that the pair $(s_j,t_j)$ is long. Hence, our instance is a no-instance if this situation occurs.

Step~6 can be performed in $O(n+k)$ time by going through the list of terminals and counting how often each terminal vertex occurs.
\qed
\end{proof} 

By Step~6, a terminal vertex may represent at most two terminals (which must belong to different terminal pairs).
We now observe that terminals should be ordered, and we let our algorithm find this ordering.

\medskip
\noindent
{\bf Step 7.} Check if there exist three terminal vertices $u$, $v$, $w$ such that $u$ and $w$ represent terminals from 
the same pair such that $l_u\leq l_v<l_w$. If so, then return a \texttt{No}-answer. Otherwise, order and rename the terminals
such that 
$r_{u_i}< l_{v_i}$  and $l_{v_{i}}\leq l_{u_{i+1}}$ for $i=1,\ldots,k'-1$, where $u_i,v_i$ are the vertices representing $s_i,t_i$, respectively.

\begin{lemma}\label{l-7}
Step 7 is safe.
\end{lemma}

\begin{proof}
We first prove that Step~7 is correct.
Suppose that there exist three terminal vertices $u$, $v$, $w$ such that $u$ and $w$ represent terminals from 
the same pair and $l_u\leq l_v<l_w$. Assume that $u$, $v$, $w$ represent $s_i$, $s_j$, $t_i$, respectively, and
let $x$ represent $t_j$.
Let $P_1$ and $P_2$ be the $s_it_i$-path and $s_jt_j$-path, respectively, in a solution. 
Because $(s_i,t_i)$ and $(s_j,t_j)$ are long, both $P_1$ and $P_2$ contain at least one inner vertex.
By Lemma~\ref{l-cover}, 
$I_{P_1}\cap I_{P_2}=\emptyset$. However, this is not possible as $l_u\leq l_v<l_w$. Hence, our instance is a no-instance.

We now show how to perform Step~7 in $O(n+k)$ time. 
Recall that each end-point of an interval is an integer between $1$ and $2n$. Construct $2n$ buckets $B_1,\ldots,B_{2n}$. Then  go through the list of terminal pairs $T$ and put a terminal in bucket $B_{l_u}$ if $u$ is the vertex of $G$ that represents the terminal.
Go through the non-empty buckets among $B_{1},\ldots,B_{2n}$ in increasing order and verify 
whether the partner of a terminal of a terminal pair not seen before is in the next non-empty bucket.
Stop and return a \texttt{No}-answer if this does not hold. Otherwise, as each bucket contains at most two terminals due to Step~6, this gives the desired ordering of the terminal pairs in $O(n+k)$ time.\qed
\end{proof}

\noindent
{\bf Step 8.}
For $i\in\{1,\ldots,k'-1\}$, if $t_i$ and $s_{i+1}$ are represented by distinct vertices $u$ and $v$,  delete all non-terminal vertices adjacent to both $u$ and $v$.

\begin{lemma}\label{l-8}
Step~8 is safe.
\end{lemma}

\begin{proof}
Any non-terminal vertex deleted in Step~8 can never be used as an inner vertex of a solution path by the definition of the \problemRIDP{} problem.
Step~8 runs in $O(n+m+k)$ time by the same arguments as in the proof of Lemma~\ref{l-3}.
\qed
\end{proof} 

\subsection{Stage II: Construct $H$}
We now construct the auxiliary $H$. Note that some intervals were already added to $H$ as part of our preprocessing stage (see Step~3).

\medskip
\noindent
{\bf Step 9.} For each $i\in\{1,\ldots,k'\}$, perform steps 9a--9d (where $u$ and $v$ are terminal vertices that represent $s_i$ and $t_i$, respectively).

\medskip
\noindent
{\bf 9a.} For every common neighbor $w$ of $u$ and $v$, add the interval $I_{uwv}$ to $H$ with color $i$, and delete $w$ from $G$.

\medskip
\noindent
{\bf 9b.} For each neighbor $x$ of $u$ not adjacent to $v$, determine whether there exists a neighbor $y$ of $v$ adjacent to $x$.
If so, then choose $y$ such that the right end-point of $y$ is leftmost amongst all such neighbours of $v$.
Add the interval $I_{uxyv}$ to $H$ with color $i$.

\medskip
\noindent
{\bf 9c.}
Determine the connected components $C_{1},\ldots,C_{p}$ of $G - (N[u] \cup N[v])$ whose vertices lie between $r_{u}$ and $l_{v}$.
For each $C_{j}$, determine the vertex $l(C_{j})$ with the leftmost left end-point and the vertex $r(C_{j})$ with the rightmost right end-point. Then among the neighbors that $l(C_{j})$ and $u$ have in common,
let $s_{i}(C_{j})$ be the one with the rightmost left end-point (if it exists). Similarly, let $t_{i}(C_{j})$ be the neighbor that $r(C_{j})$ and $v$ have in common and that has the leftmost right end-point (if it exists). 
Add the interval between the left end-point of $s_{i}(C_{j})$ and the right end-point of $t_{i}(C_{j})$ to $H$ with color $i$, if it has not been
added already in Step~9b (which might be the case if $s_i(C_j)$ and $t_i(C_j)$ intersect).

\begin{lemma}\label{l-9}
Step~9 is safe. Moreover, for $i=1,\ldots,k'$, if $P$ is a solution $s_it_i$-path, then there is a candidate $s_it_i$-path $P'$ with $I_{P'} \subseteq I_{P}$.
\end{lemma}

\begin{proof}
We first prove that Step~9 is correct.
Let $i\in \{1,\ldots,k'\}$.
Let $u$ and $v$ be the (non-adjacent) vertices of $G$ representing $s_i$ and $t_i$, respectively.
Let $P$ be a solution path for $(s_i,t_i)$. 
 
Suppose that  $P$ has length~$2$.
Then $P$ has exactly one inner vertex $w$, which is adjacent to  both $u$ and $v$.
By Step 9a, $H$ contains the interval $I_P$.
 
Suppose that $P$ has length~$3$.
Then $P$ has exactly two inner vertices $x$ and $y'$ that are adjacent to $u$ and $v$, respectively.
Let $y$ be the neighbor of $v$ that is adjacent to $x$ and has the leftmost right end-point among all such vertices. Then $P' = uxyv$ is an $s_it_i$-path. Notice that $I_{P'} \subseteq I_{P}$ by the choice of $y$ and by the fact that $u$ and $v$ have no common neighbors after Step~9a. Therefore, in any solution that contains $P$, $P$ can be replaced $P'$. By Step~9b, $H$ contains~$I_{P'}$.

Finally, suppose that $P$ has length at least~$4$. 
Because $P$ is an induced path, there is a connected component $C_j$ of $G - (N[u] \cup N[v])$ whose vertices all lie between $r_{u}$ and $l_{v}$, such that all inner vertices of $P$ except two neighbors of $u$ and $v$ are in $C_j$. 
Let $x'$ and $y'$ be the neighbors of $u$ and $v$ on $P$, respectively.
Let $x = s_{i}(C_{j})$ and $y=t_{i}(C_{j})$. Then from $P$ we can construct an $s_it_i$-path $P'$ by replacing $x'$ and $y'$ with $x$ and $y$, respectively.
Notice that $I_{P'} \subseteq I_{P}$ by the choice of $y$ and by the fact that $u$ and $v$ have no common neighbors after Step~9a. Therefore, in any solution that contains $P$, $P$ can be replaced $P'$. By Step~9c, $H$ contains~$I_{P'}$.

Observe that the above arguments prove that for $i=1,\ldots,k'$, if $P$ is a solution $s_it_i$-path, then there is a candidate $s_it_i$-path $P'$ with $I_{P'} \subseteq I_{P}$.

We now show how to perform Step~9 in $O(n+m)$ time.
In Step~9a, we add all the intervals that correspond to common neighbors of $s_{i}$ and $t_{i}$ for $i=1,\ldots,k'$, and delete these common neighbors from $G$. 
Common neighbors of $s_i$ and $t_i$ are not common neighbors of terminals of any other pair by Step~8. Therefore, Step~9a takes $O(n+m)$ time in total, and $O(n)$ intervals are added to $H$.
In Step~9b, for $i=1,\ldots,k'$, we find for each neighbor $x$ of $s_{i}$ 
(recall that $x$ is not adjacent to $t_i$ after Step~9a),
 the neighbor $y$ of $t_{i}$ such that $x$ and $y$ are adjacent and the right end-point of $y$ is leftmost. 
By using the adjacency lists for the neighbors of $u$, Step~9b takes $O(n+m)$ time in total, and $O(n)$ intervals are added to $H$.
In Step~9c, we first find the connected components $C_{1},\ldots,C_{\ell}$. This can be done by performing a breadth-first search. 
Because the connected components that we consider (and their vertices) are unique to a terminal pair, Step~9c takes $O(n+m)$ time in total. 
Again, $O(n)$ intervals are added to $H$.
\qed
\end{proof}

\subsection{Stage III: Find Independent Set}
It remains to find a particular independent set in $H$.

\medskip
\noindent
{\bf Step 10.} Find an independent set in $H$ that, for $i=1,\ldots, k$, contains exactly $r_i-1$ or $r_i$ vertices colored~$i$ depending on whether $(s_i,t_i)$ is a multi-pair or not. If such a set exists, add the corresponding candidate paths to ${\cal P}$ and return ${\cal P}$. Otherwise, return a {\tt No}-answer.

\begin{lemma}\label{l-10}
Step~10 is safe.
\end{lemma}

\begin{proof}
We first prove that Step~10 is correct.
We do this by proving that our instance is a yes-instance if and only if $H$ has an independent set as described in Step~10.
First, suppose that $H$ has such an independent set~${\cal I}$. For each interval $u$ of color $i$, we can find an $s_it_i$-path in $G$ with inner vertices that are used to construct $u$. Taking into account the terminal paths that are already included in $\cal P$, we obtain $r_i$ $s_it_i$-paths for each $i\in\{1,\ldots,k\}$.  We have to show that these paths are mutually induced. Because $\cal I$ is an independent set, distinct paths have no adjacent inner vertices. It remains to show that each $u\in {\cal I}$ does not intersect any terminal vertex (interval) of $G$ except the vertices representing $s_i,t_i$. If $u$ is added to $H$ in Step~3, then it follows immediately from the fact that all non-terminal vertices that are adjacent to at least three terminals are deleted in Step~1 and from the description of Step~3. If $u$ is added to $H$ in Step~9, then notice $u$ does not intersect any terminal vertex deleted in Step~4, because we delete them together with adjacent non-terminal vertices. Similarly, it does not interfere with any terminal deleted in Step~5, as proved in Lemma~\ref{l-5}. Moreover, each interval added in Step~9 intersects exactly two remaining terminal vertices that are partners by Step~8. Hence, the instance is a yes-instance.

Now suppose that our instance is a yes-instance. Let $\ell_i = r_i-1$ if $(s_i,t_i)$ is a multi-pair, and let $\ell_i = r_i$ otherwise.
By Observation~\ref{o-trivial}, we can assume that the solution includes all terminal paths. Therefore, the solution contains exactly $\ell_i$ $s_it_i$-path with inner vertices. By Lemma~\ref{l-3} and Lemma~\ref{l-9}, for each such solution $s_it_i$-path $P$, there is a candidate $s_it_i$ path $P'$ such that $I_{P'} \subseteq I_{P}$. Therefore, we can replace each solution path by a candidate path, and obtain a solution that uses only candidate paths. Let $\mc{I}$ denote the set of intervals covered by these paths.
By Lemma~\ref{l-cover}, the intervals of $\mc{I}$ do not intersect each other. Moreover, by construction, $\mc{I}$ contains $\ell_i$ intervals with color $i$. Therefore, $H$ has an independent set as described in Step~10.

We now show how to perform Step~10 in $O(n+m)$ time. We do this by performing the following procedure, which is a modification of the well-known greedy algorithm for finding a largest independent set in an interval graph.

\medskip
\noindent
{\bf 1.} Construct $2n$ buckets $L_1,\ldots,L_{2n}$ and $2n$ buckets $R_1,\ldots,R_{2n}$.

\medskip
\noindent
{\bf 2.} For each vertex $u$ of $H$, put $u$ in the buckets $L_{l_u}$ and $R_{r_u}$. 

\medskip
\noindent
{\bf 3.} Set ${\cal I}=\emptyset$ and $h=2n$. For $i=1,\ldots,k$, set $\ell_i=r_i-1$ if $(s_i,t_i)$ is a multi-pair, and set 
$\ell_i=r_i$ otherwise.

\medskip
\noindent
{\bf 4.}  Scan the buckets $L_{h},\ldots,L_1$ until we find a bucket $L_j$ that contains a vertex $u$ of $H$ of some color $i$ such that $\ell_i>0$. Then $u$ is included in $\cal I$. Find the set of vertices $X$ from the buckets $R_j,\ldots,R_i$, and delete them from $H$. Then  set $\ell_i=\ell_i-1$, $h=j$, and repeat the procedure. We stop as soon as we cannot find the next bucket $L_j$.   

\medskip
\noindent
If $\cal I$ contains less than $\ell_i$ vertices of color $i$ for some $i\in\{1,\ldots,k\}$, then stop and return a {\tt No}-answer.
Otherwise, return ${\cal I}$. 
This procedure takes $O(|V(H)|) = O(n)$ time, and the corresponding paths can be found in $O(n+m)$ time. Hence, it remains to show that the procedure is correct.  
We need the following claim,
which implies that between the left endpoints of two intervals with a color $i$ there can be no left endpoint of an interval with color $j \not=i$.

\medskip
\noindent
{\bf Claim 1.}
{\it Let $U_i,  U_j$ be the set of vertices (intervals) of $H$ colored by distinct colors $i$ and $j$ respectively. Then for any $u\in U_i$ and $v\in U_j$, $l_u\neq l_v$. Moreover,  if $l_u< l_v$ for some  $u\in U_i$ and $v\in U_j$, then $l_x<l_y$ for any $x\in U_i$ and $y\in U_j$.}

\medskip
\noindent
{\em Proof:}
Let $u\in U_i$ and $v\in U_j$. Suppose that $u$ and $v$ are added to $H$ in Step~3 of the algorithm.  Then $l_u\neq l_v$, because $u$ and $v$ are distinct vertices of $G$. Without loss of generality, $l_u< l_v$. 
Note that the intervals of $U_i$ correspond to the non-terminal vertices of $G$ that are adjacent to two adjacent terminal vertices $w_1,z_1$ of $G$ representing $s_i,t_i$ and that are not adjacent to other terminal vertices, by Step~1 and~3. Similarly, the intervals of $U_j$ correspond to the non-terminal vertices of $G$ that are adjacent to two adjacent terminal vertices $w_2,z_2$ of $G$ representing $s_j,t_j$ and that are not adjacent to other terminal vertices. Consider the interval $I=w_1\cap z_1$. Because  $l_u< l_v$,  the left end-point of  any $x\in U_i$ lies to the left of the right end-point of $I$ and the left end-point of  any $y\in U_j$ lies to the right of the right end-point of $I$. Hence, $l_x<l_y$ for any $x\in U_i$ and $y\in U_j$.

Suppose now that $u$ is added to $H$ in Step~3 and $v$ is added to $H$ in Step~9. The intervals of $U_i$ correspond to the non-terminal vertices of $G$ that are adjacent to two adjacent terminal vertices $w_1,z_1$ of $G$ representing $s_i,t_i$ and that are not adjacent to other terminal vertices.  The intervals of $U_j$ are the unions of non-terminal vertices of $G$ and these intervals intersect two non-adjacent terminal intervals $w_2,z_2$ of $G$ representing $s_j,t_j$. Observe that the intervals of $U_i$ could not be used for construction of the intervals of $U_j$ because all non-terminal vertices that are adjacent to $w_1,z_1$ are deleted in Steps~4 and~8. 
Moreover, the intervals of $U_j$ do not intersect any terminal vertex of $G$ except $w_2,z_2$.
Hence, $l_u\neq l_v$. Consider the interval $I=w_1\cap z_1$.
Without loss of generality, $l_u< l_v$. Then the left end-point of  any $x\in U_i$ lies to the left of the right end-point of $I$ and the left end-point of  any $y\in U_j$ lies to the right of the right end-point of $I$. Hence, $l_x<l_y$ for any $x\in U_i$ and $y\in U_j$.

Finally, suppose that $u$ and $v$ are added to $H$ in Step~9 of the algorithm. The intervals of $U_i$ intersect two non-adjacent terminal intervals $w_1,z_1$ of $G$ representing $s_i,t_i$ and they do not intersect other terminal vertices of $G$, and the intervals of $U_j$ intersect two non-adjacent terminal intervals $w_2,z_2$ of $G$ representing $s_j,t_j$ and they do not intersect other terminal vertices of $G$. Recall that the terminals are ordered in Step~7. Hence, we can assume without loss of generality that $r_{w_1}<l_{z_1}\leq l_{w_2}<r_{z_2}$. It remains to observe that each interval of $U_i$ has its left end-point to the left of $r_{w_1}$ and each interval of $U_j$ has its left end-point to the right of $r_{w_1}$. This proves Claim~1.

\medskip
Claim~1 implies that between the left endpoints of two intervals with a color $i$ there can be no left endpoint of an interval with color $j \not=i$. Then, similar as the correctness of the well-known greedy algorithm for finding a largest independent set in an interval graphs, we can argue that the above procedure outputs the required independent set.
\qed
\end{proof}

As each step in our algorithm is safe, we obtain the following result.

\begin{theorem}\label{t-interval}
The \problemRIDP{} problem can be solved in time 
$O(n+m+k)$ for interval graphs on $n$ vertices and $m$ edges with $k$ terminal pairs.
\end{theorem}

\section{Circular-Arc Graphs}\label{s-circ}
In this section, we modify the algorithm of the previous section to work for the \problemIDP{} problem on circular-arc graphs. The general idea of the approach remains the same, but some preprocessing steps are no longer needed, and some steps need modification. In particular, we do not need colors here. We will again show that each step of the algorithm is \emph{safe}, where the definition of a safe step remains the same, mutatis mutandis. The algorithm assumes that an arc representation of $G$ is known, as given by Theorem~\ref{t-mcc}. It maintains an auxiliary circular-arc graph $H$, initially empty, in a similar manner and function as before. It also maintains a set $\mc{P}$ of paths, initially empty.

The algorithm first performs Step~1. Note that Step~2 and~3 are not necessary, as there are no multi-pairs now, and thus we do not apply them. We then continue with Step~4 and~5.

\begin{lemma} \label{l-c-145}
Step~1, 4, and~5 are safe.
\end{lemma}
The proof of this lemma is obtained in the same way as the proofs of Lemmas~\ref{l-1}, \ref{l-4}, and~\ref{l-5}.

After Step~5, for each remaining terminal pairs $(s_i,t_i)$, $s_i$ and $t_i$ are represented by vertices at distance at least two, and as before, we call such pairs \emph{long}. Let $k'$ be the number of remaining terminal pairs. Notice that it can happen that $k'\leq 1$ after Step~5. It is convenient to handle this case separately.

\newcommand{\sfiv}{5$^{+}$}

\medskip
\noindent
{\bf Step \sfiv{}.} If $k'=0$, then stop and return the solution $\cal P$. If $k'=1$, then consider the terminal vertices $u$ and $v$ representing the terminals of the unique pair of $T$. Find a shortest $uv$-path $P$ if it exists. If $P$ exists, then add $P$ to $\cal P$, and return the solution $\cal P$. Otherwise, stop and return a {\tt No}-answer.

\begin{lemma}\label{l-iv}
Step~\sfiv{} is safe.
\end{lemma}

\begin{proof}
It is clear that Step~\sfiv{} can be executed in $O(n+m)$ time. The cases that $k'=0$ and that $k'=1$ and $P$ does not exist are trivially correct. If $k'=1$ and $P$ does exist, then $P$ cannot have any inner (non-terminal) vertices that are adjacent to the terminal vertices that are deleted in Step~4, because any such non-terminal vertices are deleted as well. Moreover, $P$ cannot have any inner (non-terminal) vertices that are adjacent to the terminals that are deleted in Step~5, as any such non-terminal vertex would either be adjacent to three terminals and thus removed in Step~1, or be adjacent to a terminal vertex of the single remaining terminal pair.
\qed
\end{proof}

Now we can assume that $k' \geq 2$. Since all pairs are long and $k'\geq 2$, there is only one direction around the circle that a solution path can go, and therefore, intuitively, the problem starts to behave roughly as it does on interval graphs. We perform Step~6, 7, 8, and~9, where in Step~9 we do not color the vertices.

\begin{lemma}\label{l-c-6789}
Steps~6, 7, 8, and~9 are safe. Moreover, for $i=1,\ldots,k'$, if $P$ is a solution $s_it_i$-path, then there is a candidate $s_it_i$-path $P'$ with $I_{P'} \subseteq I_{P}$.
\end{lemma}

\begin{proof}
The lemma follows immediately from Lemmas~\ref{l-6}, \ref{l-7}, \ref{l-8}, and~\ref{l-9}. Notice that in the proof of Lemma~\ref{l-7}, we need to be slightly careful: if the first two non-empty buckets contain terminals from different terminal pairs, then since we are dealing with circular-arc graphs, this does not immediately mean that we should return a \texttt{No}-answer. Instead, we should restart the procedure with the second non-empty bucket, and move the first non-empty bucket to the end of the list (as bucket $B_{2n+1}$).
\qed\end{proof}

Finally, we execute the following simplified version of Step~10.

\newcommand{\sten}{10$^{*}$}

\medskip
\noindent
{\bf Step \sten{}.}  Find a largest independent set in $H$ using Theorem~\ref{t-maxin}. If such a set exists, add the corresponding candidate paths to $\mc{P}$ and return $\mc{P}$. Otherwise, return a {\tt No}-answer.

\begin{lemma}\label{l-c-10}
Step~\sten{} is safe.
\end{lemma}

\begin{proof}
A largest independent set can be found in $O(n)$ time using Theorem~\ref{t-maxin}. Then the corresponding paths can be found in $O(n+m)$ time. To prove that Step~\sten{} is correct, we prove that the instance is a yes-instance if and only if $H$ has an independent set of size at least $k'$.

Suppose that $\cal I$ is an independent set of $H$ of size at least $k'$. By the construction of $H$, the set of vertices of $H$ can be partitioned into $k'$ sets $X_1,\ldots,X_{k'}$ such that for each $i\in\{1,\ldots,k'\}$, $X_i$ contains only intervals that intersect the vertices $u,v$ representing $s_i,t_i$, respectively, in $r_u$ and $l_v$. Hence, $\cal I$ has exactly one vertex from each $X_1,\ldots,X_{k'}$.   
For each interval $w$ in $\cal I$ from $X_i$, we can find an $s_it_i$-path in $G$ with inner vertices that are used to construct $w$. Taking into account the paths that are already included in $\cal P$, we obtain $s_it_i$-paths for each $i\in\{1,\ldots,k\}$.  We have to show that these paths are mutually induced. Because $\cal I$ is an independent set, distinct paths 
have no adjacent inner vertices. It remains to show that each $w\in {\cal I}$ does not intersect any terminal vertex (interval) of $G$ except the vertices representing $s_i,t_i$.
 Notice that $w$ does not intersect any terminal vertex deleted in Step~4, because we delete them together with adjacent non-terminal vertices. Similarly, as argued in Lemma~\ref{l-5}, $w$ does not interfere with any terminals deleted in Step~5.
Recall that non-terminal vertices that are adjacent  to at least three distinct terminal vertices are deleted in Step~1. By Step~8 and the fact that the common neighbors of two terminals are deleted in the first phase of the construction of $H$ in Step~9a, we obtain that $w$ does not intersect any terminal except $s_i,t_i$. Hence, the instance is a yes-instance.

Suppose now that we have a yes-instance of \problemIDP{} and consider a solution to the instance. By Observation~\ref{o-trivial}, we can assume that the solution includes all terminal paths from $\cal P$. We consider remaining $k'$ paths that have inner vertices. By Lemma~\ref{l-c-6789}, for each solution $s_it_i$-path $P$, there is a candidate $s_it_i$-path with $I_{P'} \subseteq I_{P}$. Hence, we may assume that each solution path is a candidate path. 
Let $\cal I$ be the set of intervals covered by these paths.
Because the paths are mutually induced, the intervals of $\cal I$ do not intersect each other. Hence, $H$ has an independent set of size $k'$.
\qed
\end{proof}

As each step in our algorithm is safe, we obtain the following result.
\begin{theorem}\label{t-circ}
The \problemIDP{} problem can be solved in time 
$O(n+m+k)$ for circular-arc graphs on $n$ vertices and $m$ edges with $k$ terminal pairs.
\end{theorem}

\section{Conclusion}\label{s-con}
We gave a linear-time algorithm for the \problemRIDP{} problem on interval graphs,
and for the \problemIDP{} problem on circular-arc graphs. It can be observed that by the application of the same ideas, we can solve 
\problemRIDP{} on $n$-vertex circular-arc graphs in time $O(n^2)$.  
We leave it as an open question, whether \problemRIDP{} can be solved in linear time for this graph class.

Another interesting question is whether the multicolored independent set problem that we solve in Step~10 of the algorithm can be solved in polynomial time on interval graphs when no order on the colors is known. In the appendix, we answer this question negatively.

\appendix

\section{\problemMCIS}
In Step~10 of the algorithm for interval graphs, we solve an instance of a generalization of the following problem:
\begin{center}
\begin{boxedminipage}{1\textwidth}
\problemMCIS{}\\[3mm]
\begin{tabular}{r l }
\textit{Instance:} & a graph $G$, an integer $k$, and a function $c : V(G) \rightarrow \{1,\ldots,k\}.$\\
\textit{Question:} & does $G$ have an independent set $I$ with $\bigcup_{v \in I} c(v) = \{1,\ldots,k\}$?  \\
\end{tabular}
\end{boxedminipage}
\end{center}
In Step~10, we essentially show that such an instance can be solved in polynomial time on interval graphs if for any two vertices $u,w$ with $c(u)=c(w)=i$ there is no vertex $v$ with $c(v)=j$ and $l_u < l_v < l_w$. However, on general interval graphs, this problem becomes NP-complete.

\begin{theorem}
\problemMCIS{} on interval graphs is NP-complete.
\end{theorem}
\begin{proof}
We show in fact that the problem is already NP-complete on disjoint unions of double stars (\ie graphs obtained from two disjoint stars by joining the central vertices), which form a subclass of interval graphs. We reduce from \textsc{3-SAT}. Consider an instance of \textsc{3-SAT} with $n$ variables $x_1,\ldots,x_n$ and $m$ clauses $C_1,\ldots,C_m$. We construct a graph $G$ and a function $c$ as follows. For each $x_i$, we create two adjacent vertices $x_i$ and $\bar{x}_i$ with $c(x_i) = c(\bar{x}_{i}) = i$. For each $C_j$, we create three vertices and set $c(\cdot)$ of these vertices to $j+n$. We then make these three vertices adjacent to the corresponding literal vertices (for example, if $C_j$ contains $x_i,\bar{x}_j,x_l$, then we join the first vertex with the vertex $x_i$, the second with $\bar{x}_j$ and the third with $x_l$). This completes the construction. Note that it is indeed a disjoint union of double stars. The correctness can be seen as follows: we set $x_i$ to true if and only if the vertex $x_i$ is not in the independent set.
\qed\end{proof}
It is easy to show that \problemMCIS{} is fixed-parameter tractable on interval graphs: guess an ordering of the colors, and for each choice, run a procedure similar to the one described for Step~10. A faster algorithm can be obtained using dynamic programming.
\end{document}